\documentclass[11pt]{article}
\usepackage{amsmath,amssymb,amsbsy,amsfonts,amsthm,latexsym,
               amsopn,amstext,amsxtra,euscript,amscd,color}
\usepackage{graphicx}
\usepackage{float}

\newcommand{\fun}[3]{{{#1}\,:\,{#2}\,\rightarrow\,{#3}}}

\newcommand{\Z}{{\mathbb Z}}

\newcommand{\F}{{\mathbb F}}

\newcommand{\U}{{\mathbb U}}

\newcommand{\cB}{{\mathcal B}}

\newcommand{\beq}{\begin{equation}}

\let\workingver=n
\def\begeq#1{\begin{equation}\mylabel{#1}}
\def\endeq{\end{equation}}
\def\begalg{\begin{alg}}
\def\endalg{\end{alg}}

\def\refeq#1{\if\workingver y(\ref{#1})-[[#1]]\else(\ref{#1})\fi}
\def\refth#1{\if\workingver y\ref{#1}-[[#1]]\else\ref{#1}\fi}
\def\mylabel#1{\if\workingver y\label{#1}{\bf\ \ [[#1]]\ \ }
\else\label{#1}\fi}
\def\mybibitem#1{\if\workingver y\bibitem{#1}{\bf\ \ [[#1]]\ \ }
\else\bibitem{#1}\fi}

\def\uu{{\mathbf u}}
\def\ww{{\mathbf w}}
\def\vv{{\mathbf v}}
\def\xx{{\mathbf x}}
\def\yy{{\mathbf y}}
\def\zz{{\mathbf z}}
\def\aa{{\mathbf a}}
\def\bb{{\mathbf b}}
\def\cc{{\mathbf c}}
\def\00{{\mathbf 0}}

\newcommand{\V}{\mathbb{V}}
\newcommand{\BBF}{\mathbb{F}}

\def\cB{{\mathcal B}}
\def\cC{{\mathcal C}}

\def\cH{{\mathcal H}}

\def\cW{{\mathcal W}}

\def\cGB{\mathcal{GB}}

\def\aa{{\bf a}}
\def\bb{{\bf b}}
\def\cc{{\bf c}}
\def\dd{{\bf d}}
\def\uu{{\bf u}}
\def\vv{{\bf v}}
\def\ww{{\bf w}}
\def\xx{{\bf x}}
\def\yy{{\bf y}}
\def\zz{{\bf z}}

\def\00{{\bf 0}}
\def\11{{\bf 1}}
\def\+{\oplus}

\def \F {{\mathbb F}}

\def \Z {{\mathbb Z}}

\def \V {{\mathbb V}}

\begin{document}
\newtheorem{thm}{Theorem}
\newtheorem{rem}[thm]{Remark}
\newtheorem{lem}[thm]{Lemma}
\newtheorem{exa}[thm]{Example}
\newtheorem{theorem}[thm]{Theorem}
\newtheorem{alg}[thm]{Algorithm}
\newtheorem{cor}[thm]{Corollary}
\newtheorem{conj}[thm]{Conjecture}
\newtheorem{prop}[thm]{Proposition}
\newtheorem{heu}[thm]{Heuristic}

\theoremstyle{definition}

\newtheorem{ex}[thm]{Example}

\theoremstyle{remark}
\newtheorem*{remark}{Remark}

\newtheorem*{ack}{Acknowledgement}

\title{\bf Generalized bent Boolean functions and strongly regular Cayley graphs  }
\author{Constanza Riera$^1$,  Pantelimon St\u anic\u a$^2$, Sugata Gangopadhyay$^3$,  \vspace{.4cm}\\
 $~^1$Department of Computing, Mathematics, and Physics,\\
 Western Norway University of Applied Sciences, \\
5020 Bergen, Norway; {\tt csr@hvl.no}\\
$~^2$Department of Applied Mathematics, \\
Naval Postgraduate School,\\
Monterey, CA 93943, USA; {\tt pstanica@nps.edu}\\
$~^3$Department of Computer Science and Engineering\\
Indian Institute of Technology Roorkee,\\
 Roorkee 247667, INDIA; {\tt gsugata@gmail.com}
}


\date{November 21, 2017}

\maketitle

\baselineskip=1.1\baselineskip

\begin{abstract}
In this paper we define the (edge-weighted) Cayley graph  associated to a generalized Boolean function, introduce a notion of strong regularity and give several of its properties.  We show some connections between this concept and generalized bent functions (gbent), that is, functions with flat Walsh-Hadamard spectrum. In particular, we find a complete characterization of quartic gbent functions in terms of the strong regularity  of their associated Cayley graph.
\end{abstract}


\section{(Generalized) Boolean functions background}
\label{intro}

Let $\V_n$ be the vector space of dimension
$n$ over the two element field $\BBF_2$, and for a positive integer $q$, let $\Z_q$ be the ring of integers modulo $q$.
Let us denote the addition, respectively, product operators over $\BBF_2$ by ``$\oplus$'', respectively, ``$\cdot$''.
A Boolean function $f$ on $n$ variables is a mapping
from $\V_n$ into $\BBF_2$,
that is, a multivariate polynomial over $\BBF_2$,
\begin{equation}
\label{eq:ANF}
f(x_1, \ldots, x_n)=a_0 \oplus \sum_{i=1}^{n} a_i x_i \oplus
\sum_{1 \leq i < j \leq n} a_{ij} x_i x_j \oplus  \ldots \oplus
a_{12\ldots n} x_1 x_2 \ldots x_n,
\end{equation}
where the coefficients
$a_0, a_i, a_{ij}, \ldots, a_{12\ldots n} \in \BBF_2$.
This representation of $f$ is called the {\em algebraic normal form} (ANF) of
$f$. The number of variables in the highest order product term with
nonzero coefficient is called the {\em algebraic degree}, or simply the
degree of $f$.

{F}or a Boolean function on $\V_n$,   the {\em Hamming weight} of $f$, $wt(f)$, is the cardinality of $\Omega_f=\{\xx\in \V_n\,:\, f(\xx)=1 \}$ (this is extended to any vector, by taking its weight to be the number of nonzero components of that vector).
 The {\em Hamming distance} between two
functions $f,\, g:\V_n\to \F_2$ is $d(f,g)=wt(f\oplus g)$.
A Boolean function $f(\xx)$   is called an {\em affine function} if its algebraic degree is 1.
If, in addition, $a_0=0$ in \refeq{eq:ANF}, then $f(\xx)$ is  a {\em linear function} (see~\cite{CS17} for more on Boolean functions).
In $\V_n=\F_2^n$, the vector space of the $n$-tuples over $\F_2$, we use the conventional dot product $\uu\cdot\xx$ as an inner product.


For a {\it generalized Boolean function} $f:\V_n\to \Z_q$ we define the {\it generalized Walsh-Hadamard transform} to be the complex valued function
\[ 
\mathcal{H}^{(q)}_f(\uu) = \sum_{\xx\in \V_n}\zeta_q^{f(\xx)}(-1)^{\uu\cdot\xx}, 
\]
where $\zeta_q = e^{\frac{2\pi \imath}{q}}$ 
(we often use $\zeta$, $\cH_f$, instead of $\zeta_q$, respectively, $\cH_f^{(q)}$,  when $q$ is fixed).
The inverse is given by $\zeta^{f(\xx)}=2^{-n}\sum_{\uu} \cH_f(\uu)(-1)^{\uu\cdot\xx}$.
For $q=2$, we obtain the usual {\it Walsh-Hadamard transform}
\[
\mathcal{W}_f(\uu) = \sum_{\xx\in \V_n}(-1)^{f(\xx)}(-1)^{\uu\cdot\xx},
\]
which defines the coefficients of character form of $f$ with respect to the orthonormal basis of the
group characters
$\chi_\ww(\xx)=(-1)^{\ww\cdot \xx}$. In turn, $f(\xx)=2^{-n} \sum_\ww \cW_f(\ww) (-1)^{\uu\cdot\xx}$.

 We use the notation as in \cite{mmms17,mms0,mms,smgs,TXQF16} (see also~\cite{ST09,Tok}) and denote the set of all generalized Boolean functions by $\mathcal{GB}_n^q$ and when $q=2$, by $\mathcal{B}_n$.
A function $f:\V_n\rightarrow\Z_q$ is called {\em generalized bent} ({\em gbent}) if $|\mathcal{H}_f(\uu)| = 2^{n/2}$ for all $\uu\in \V_n$.
We recall that a function $f$ for which $|\mathcal{W}_f(\uu)| = 2^{n/2}$ for all $\uu\in \V_n$ is called a {\em bent} function, which only exist for even $n$ since
$\mathcal{W}_f(\uu)$ is an integer.
Let $f\in \mathcal{GB}_n^{q}$, where $2^{k-1}< q\leq 2^k$, then we can represent $f$ uniquely as
\[ f(\xx) = a_0(\xx) + 2a_1(\xx) + \cdots + 2^{k-1}a_{k-1}(\xx) \]
for some Boolean functions $a_i$, $0\le i\le k-1$ (this representation  comes from the binary representation of the elements in the image set $\Z_{2^k}$). For results on classical bent functions and related topics, the reader can consult~\cite{Bud14,CS17,MesnagerBook,Tok15}.


 \section{Unweighted strongly  regular graphs}

A graph is {\em regular of degree $r$}\index{Graph@!regular} (or $r$-regular) if
every vertex has degree $r$
(number of edges incident to it).
We say that an $r$-regular graph $G$
is a {\em strongly regular graph} (srg)\index{Graph@!strongly regular, srg} with parameters $(v,r,e,d)$ if
 there exist nonnegative
integers $e,d$ such that for all vertices $\uu,\vv$
the number of vertices adjacent to both $\uu,\vv$ is $d$, $e$, if $\uu,\vv$ are
adjacent, respectively, nonadjacent (see for instance \cite{CDS}).
The {\em complementary} graph $\bar G$ of the strongly regular graph
$G$ is also
strongly regular with parameters
$(v,v-r-1,v-2r+e-2,v-2r+d)$ (see \cite{CDS}).

Since the objects of this paper are edge-weighted graphs $G=(V,E,w)$ (with vertices $V$, edges $E$ and weight function $w$ defined on $E$ with values in some set, which in our case it will be either the set of integers modulo $q$, $\Z_q$ with $q=2^k$, or the complex numbers set $\mathbb{C}$), we define the {\em weighted degree} $d(v)$ of a vertex $v$ to be the sum of  the weights of its incident edges, that is, $\displaystyle d(v)=\sum_{u, (u,v)\in E} w(u,v)$ (later, we will introduce yet another degree or strength concept). Certainly, one can also define the {\em combinatorial degree} $r(v)$ of a vertex to be the number of such incident edges.
For more on graph theory the reader can consult~\cite{Biggs,CDS} or one's favorite graph theory book.

Let $f$ be a Boolean function on $\V_n$.
We define the {\em Cayley graph} of $f$ to be the graph $G_f=(\V_n,E_f)$
 whose vertex set is $\V_n$ and the set of edges is defined by
\[
E_f=\{(\ww,\uu)\in \V_n\times \V_n\,:\, f(\ww\oplus \uu)=1\}.
\]

 For some fixed (but understood from the context) positive integer $s$, let the  canonical injection  $\fun{\iota} {\V_s}{\Z_{2^s}}$  be defined by
$\iota(\cc) = \cc\cdot (1,2,\ldots,2^{s-1})=\sum_{j=0}^{s-1}c_j 2^j$, where $\cc=(c_0,c_1,\dots,c_{s-1})$. For easy writing, we denote by ${\bf j}:=\iota^{-1}(j)$.

The adjacency matrix $A_f$ is the matrix whose entries are $A_{i,j}=f({\bf i}\oplus {\bf j})$ (here $\iota$ is defined on $\V_n$).
It is simple to prove that $A_f$ has the dyadic property: $A_{i,j}=A_{i+2^{k-1},j+2^{k-1}}$.
Also, from its definition, we derive that $G_f$ is a {\em regular graph of degree
$wt(f)=|\Omega_f|$}
(see \cite[Chapter 3]{CDS} for further definitions).

Given a graph $f$ and its adjacency matrix $A$, the {\em spectrum}, with notation $Spec(G_f)$, is the
 set of eigenvalues of $A$ (called also the eigenvalues of $G_f$).
We assume throughout that $G_f$ is connected (in fact, one can show that all connected components of $G_f$ are isomorphic).

  It is known (see \cite[pp. 194--195]{CDS})
that a connected $r$-regular graph is strongly regular iff it has exactly
three distinct eigenvalues
$\lambda_0=r,\lambda_1,\lambda_2$ (so $e=r+\lambda_1\lambda_2+\lambda_1+\lambda_2$,
$d=r+\lambda_1\lambda_2$).

The following result is known \cite[Th. 3.32, p. 103]{CDS} (the second part follows from a counting argument and is also well known).
\begin{prop}
\label{prop1}
The following identity holds for a strongly $r$-regular graph:
\[
A^2=(d-e)A+(r-e)I+eJ,
\]
where $J$ is the all $1$ matrix. Moreover,
$
r(r-d-1)=e(v-r-1).
$
\end{prop}

In~\cite{BC1,BCV} it was shown that a Boolean function $f$ is  bent   if and only if the Cayley graph $G_f$ is strongly regular with $e=d$. We shall refer to this as the Bernasconi-Codenotti correspondence.

\section{The Cayley graph of a generalized Boolean function}

We now let $f:\V_n\to\Z_q$ be a generalized Boolean function. We define the {\em (generalized) Cayley graph} $G_f$ to be the graph where vertices are the elements of $\V_n$ and two vertices $\uu,\vv$ are connected by a weighted edge of (multiplicative) weight $\zeta^{f(\uu\oplus\vv)}$ (respectively, additive weight $f(\uu\oplus\vv)$). Certainly, the underlying unweighted graph is a complete pseudograph (every vertex also has a  loop). We sketch in Figure~\ref{gcayley1} such an example.

\begin{figure}[tbh!]
\begin{center}
\includegraphics[width=.8\textwidth]{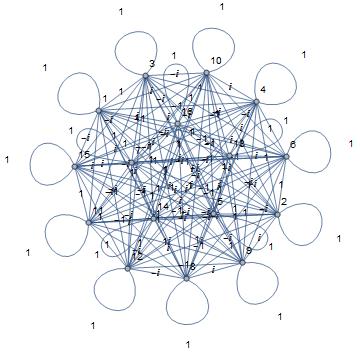}
\end{center}
\caption{Cayley graph associated to the gbent $f(\xx)=x_1  +2 (x_1 x_2 \oplus x_3 x_4 )$}
\label{gcayley1}
\end{figure}

Certainly, one can define a modified (generalized) Cayley graph $G'_f$ where two vertices are connected if and only if $f(\uu\oplus\vv)\neq 0$ with weights given by $\zeta^{f(\uu\oplus\vv)}$. We sketch in Figure~\ref{mcayley1} such a graph (it is ultimately the above graph with all weight 1 edges removed).
\begin{figure}[tbh!]
\begin{center}
\includegraphics[width=.6\textwidth]{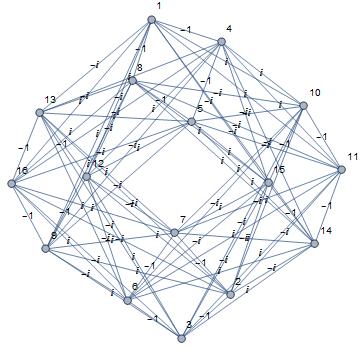}
\end{center}
\caption{Modified Cayley graph associated to gbent $f(\xx)=x_1  +2 (x_1 x_2 \oplus x_3 x_4 )$}
\label{mcayley1}
\end{figure}

In Example~\ref{exp2}, we give an example of a generalized Cayley graph, and its spectrum.
\begin{exa}
 \label{exp2} Let $f:{\mathbb V_n}\rightarrow {\mathbb Z}_4$ defined by $f(x_1,x_2)=x_1x_2+2x_1$. The truth table is $(0\ 0\ 2\ 3)^T$ $($using the lexicographical order $x_1, x_2$$)$. Then, the adjacency matrix (with multiplicative weights) is
$$ A_f=\left(\begin{array}{cccc}
1&1&-1&-i\\
1&1&-i&-1\\
-1&-i&1&1\\
-i&-1&1&1
\end{array}\right).$$
A basis for its eigenspace is $\{\vec{v}_1,\vec{v}_2,\vec{v}_3,\vec{v}_4\}$, where $\vec{v}_1=(1\ 1 \ 1\ 1)^T$ with $\chi_1(\xx)=(-1)^0$ $\vec{v}_2=(1\ -1 \ 1\ -1)^T$ with $\chi_2(\xx)=(-1)^{x_2}$, $\vec{v}_3=(1\ 1 \ -1\ -1)^T$ with $\chi_3(\xx)=(-1)^{x_1}$, $\vec{v}_4=(1\ -1 \ -1\ 1)^T$ with $\chi_4(\xx)=(-1)^{x_1+x_2}$, having respective eigenvalues $\lambda_0=1-i,\,\lambda_1=-1+i,\,\lambda_2=3+i,\,\lambda_3=1-i$.
We can see that the eigenvalues $A_f$ are
\begin{align*}
\lambda_0&=i^0\chi_1(00)+i^0\chi_1(01)+i^2\chi_1(10)+i^3\chi_1(11)=1+1+i^2+i^3=1-i={\cal{H}}_f^{(4)}({\bf 0}),\\
\lambda_1&=i^0\chi_2(00)+i^0\chi_2(01)+i^2\chi_2(10)+i^3\chi_2(11)=1-1+i^2-i^3=-1+i={\cal{H}}_f^{(4)}({\bf 1}),\\
\lambda_2&=i^0\chi_3(00)+i^0\chi_3(01)+i^2\chi_3(10)+i^3\chi_3(11)=1+1-i^2-i^3=3+i={\cal{H}}_f^{(4)}({\bf 2}),\\
\lambda_3&=i^0\chi_4(00)+i^0\chi_4(01)+i^2\chi_4(10)+i^3\chi_4(11)=1-1-i^2+i^3=1-i={\cal{H}}_f^{(4)}({\bf 3}).
\end{align*}
\end{exa}

Although, we do not use it in this paper, we define the {\em strength} of the vertex $\aa$ in the Cayley graph $G_f$  as the sum of the additive weights of incident edges, that is, $s(\aa)=\sum_b f(\aa\oplus \bb)$.

\begin{rem}
If $f\in\cGB_n^q$ and $G_f$ is its Cayley graph, we observe that all vertices are adjacent of multiplicative $($respectively, additive$)$ weights  in $\U_q=\{1,\zeta,\zeta^2,\ldots,\zeta^{q-1}\}$ $($respectively, in $\mathbb{Z}_q=\{0,1\ldots,q-1\}$$)$.
\end{rem}

We next show that the eigenvalues of the Cayley graph $G_f$ (with multiplicative weights) are precisely the (generalized) Walsh-Hadamard coefficients.
\begin{thm}
Let $f:\V_n\to\Z_q$, $q=2^k$, and let $\lambda_i, 0\leq i\leq 2^n-1$ be the eigenvalues of its associated (multiplicative) edge-weighted graph $G_f$. Then,
\[
\lambda_i=  \cH_f^{(q)}({\bf i})\ (\text{recall that  ${\bf i}=\iota^{-1}(i)$}).
\]
\end{thm}
\begin{proof}
 Let $\chi:{\mathbb V}_n\rightarrow {\mathbb C}$ be a character of $\V_n$, and for each such character, let ${\bf x}_\chi=\left(x_j\right)_{0\leq j\leq 2^n-1}\in{\mathbb C}^{2^n}$, where $x_j=\chi({\bf j})$. We claim (and show) that ${\bf x}_\chi$ is an eigenvector of $A=A_f$ (for simplicity, we use $A$ in lieu of $A_f$ in this proof), with eigenvalue $\displaystyle \sum_{k=0}^{q-1}\sum_{{\bf s_k}\in S_k}\zeta^k \chi({\bf s_k})$, where $S_k=\{{\bf s}_k: f({\bf s}_k)=k\}$.
(Observe that the characters of ${\mathbb V}_n$ are $\chi_{{\bf w}}({\bf x})=(-1)^{\uu\cdot\xx}$, and thus the eigenvalues are exactly the Walsh--Hadamard transform coefficients).

The $i$-th entry of $A{\bf x}$ is
$$(A{\bf x})_i=\sum_j A_{i,j}x_j=\sum_j A_{i,j}\chi({\bf  j})=\sum_{k=0}^{q-1}\sum_{{\bf i}\oplus{\bf j}\in S_k}\zeta^k\chi({\bf j})$$
If ${\bf i}\oplus{\bf j}\in S_k$, then ${\bf i}\oplus{\bf j}={\bf s}_k$, for some ${\bf s}_k\in S_k$, and so, ${\bf j}={\bf i}\oplus{\bf s}_k$. Since $\chi$ is a character, $$\chi({\bf j})=\chi({\bf i}\oplus{\bf s}_k)=\chi({\bf i})\chi({\bf s}_k)=x_i\chi({\bf s}_k)$$
Then,
$$(A{\bf x})_i=\sum_{k=0}^{q-1}\sum_{{\bf s}_k\in S_k}\zeta^k x_i\chi({\bf s}_k)=x_i\sum_{k=0}^{q-1}\sum_{{\bf s}_k\in S_k}\zeta^k \chi({\bf s}_k),$$
which shows our theorem.
\end{proof}

%
%

\section{Generalized bents and their Cayley graphs}

We recall that a {\em $q$-Butson Hadamard matrix}~\cite{Bu62} ($q$-BH) of dimension $d$ is a $d \times d$
matrix $H$ with all entries $q$-th roots of unity such that
$HH^* = d I_d$, where $H^*$ is the conjugate transpose of $H$. When $q = 2$, $q$-BH matrices are called Hadamard matrices (where the
entries are $\pm 1$).
Recall that the {\em crosscorrelation} function is defined by $$\cC_{f,g}(\zz)=\sum_{\xx \in \V_n} \zeta^{f(\xx)  - g(\xx \+ \zz)},$$
and the {\em autocorrelation} of $f \in \cGB_n^q$ at $\uu \in \V_n$
is $\cC_{f,f}(\uu)$ above, which we denote by $\cC_f(\uu)$.

\begin{thm}
Let $f\in\cGB_n^q$. Then  $f$ is gbent if and only if the adjacency matrix $A_f$ of the (multiplicative) edge-weighted  Cayley graph associated to  $f$ is a $q$-Butson Hadamard matrix.
\end{thm}
\begin{proof}
Let $A_f=\left(\zeta^{f(\aa+\bb)} \right)_{\aa,\bb}$. Then, the $(\aa,\bb)$-entry of $A_f\cdot \bar A_f$ is
\begin{align}
\label{eq2}
\left(A_f\cdot \bar A_f\right)_{\aa,\bb}=\sum_{\cc\in\V_n} \zeta^{f(\aa\+\cc)}\zeta^{\bar f(\cc\+\bb)}=\sum_{\cc\in\V_n} \zeta^{f(\aa\+\cc)- f(\cc\+\bb)}=\cC_{f}(\aa\+\bb).
\end{align}
Now, recall from~\cite{smgs} that if $f,g\in\cGB_n^q$, then
\begin{equation*}
\label{eq3}
\begin{split}
\sum_{\uu \in \V_n}\cC_{f,g}(\uu)(-1)^{\uu\cdot\xx} = 2^{-n} \cH_f(\xx)\overline{\cH_g(\xx)}, \\
\cC_{f,g}(\uu) = 2^{-n}\sum_{\xx \in \V_n}\cH_f(\xx)\overline{\cH_g(\xx)} (-1)^{\uu\cdot\xx}.
\end{split}
\end{equation*}
Thus, equation~\eqref{eq2} becomes
\[
\left(A_f\cdot \bar A_f\right)_{\aa,\bb}=2^{-n}\sum_{\xx \in \V_n}\cH_f(\xx)\overline{\cH_f(\xx)} (-1)^{(\aa\+\bb)\cdot \xx}=
2^{-n}\sum_{\xx \in \V_n} \|\cH_f(\xx)\|^2 (-1)^{(\aa\+\bb)\cdot \xx}.
\]
By Parseval's identity,  if $\aa=\bb$, then $\left(A_f\cdot \bar A_f\right)_{\aa,\aa}=2^n$. Assume now that $\aa\neq \bb$ and we shall show that  $\left(A_f\cdot \bar A_f\right)_{\aa,\bb}=0$ for some $\aa\neq\bb$ if and only if $f$ is gbent. Certainly, if $f$ is gbent then $\|\cH_f(\xx)\|^2=2^n$, and since $\sum_{\xx \in \V_n}  (-1)^{(\aa\+\bb)\cdot \xx}=0$, for $\aa\+\bb\neq 0$, we have that implication.
We can certainly show it directly, but the converse follows from~\cite[Theorem 1 ($iv$)]{smgs}.
\end{proof}

{\em For the remaining of the paper, for simplicity, we shall only consider additive weights, namely, our edge-weighted graphs $(V,E,w)$ will have the weight function $w:E\to\mathbb{Z}_q$, $q=2^k$.}

Next, we say that a weighted graph $G=(V,E,w)$, $V\subseteq \V_n$, $w:E\to \Z_q$, $q=2^k$, is a  {\em weighted regular graph} (wrg) of parameters $(v;r_0,r_1,\ldots,r_{q-1})$ if every vertex will have exactly $r_j$ neighbors of edge weight $j$. We denote by $N_j(\aa)$ the set of all neighbors of a vertex $\aa$ of corresponding edge weight $j$.

\begin{prop}
\label{prop6}
Given a generalized Boolean function $f\in\cGB_n^q$, the associated Cayley graph is weighted regular (of some parameters), that is, every vertex will have the same number of incident edges  with a fixed weight.
\end{prop}
\begin{proof}
Fix a weight $j$ and a vertex $\xx_0$, and consider the equation $f(\xx_0\oplus \yy)=j$ with solutions $\yy_1,\yy_2,\ldots,\yy_t$, say. For any other vertex $\xx_1$, the equation $f(\xx_1\oplus \yy)=j$ will have solutions $\yy_1\oplus \xx_1\oplus \xx_0,\yy_2\oplus \xx_1\oplus \xx_0, \ldots, \yy_t\oplus \xx_1\oplus \xx_0$. The proof of the lemma is done.
\end{proof}

We will define our first concept of strong regularity here.
Let   $X,\bar X$ be a fixed bisection of  the weights $\Z_q=X\cup \bar X, X\cap \bar X=\emptyset, |X|=|\bar X|=2^{k-1}$, and let $Y\subseteq \Z_q$. We say that a weighted regular (of parameters $(v;r_0,r_1,\ldots,r_{q-1})$) graph $G=(V,E,w)$, $V\subseteq \V_n$, $w:E\to \Z_q$, $q=2^k$,  is a (generalized) {\em $(X;Y)$-strongly regular} (srg) of parameters $(v;r_0,r_1,\ldots,r_{q-1};e_X,d_X)$ if and only if the number of vertices $\cc$ adjacent to both  $\aa,\bb$, with $w(\aa,\cc)\in Y,w(\bb,\cc)\in Y$, is exactly $e_X$ if  $w(\aa,\bb)\in X$, respectively, $d_X$ if  $w(\aa,\bb)\in \bar X$. One can weaken the condition and define a $(X_1,X_2;Y)$-srg notion, where $X_1\cap X_2=\emptyset$, not necessarily a bisection, and require the number of vertices $\cc$ adjacent to both  $\aa,\bb$, with $w(\aa,\cc)\in Y,w(\bb,\cc)\in Y$, to be exactly $e_X$ if  $w(\aa,\bb)\in X_1$, respectively, $d_X$ if  $w(\aa,\bb)\in  X_2$; or even
  allowing a multi-section, and all of these variations can be fresh areas of research for graph theory experts.

 Note that our definition (see also~\cite{CJMPW} for an alternative concept, which we mention in the last section) is a natural extension of the classical definition: Let $q=2$, and $X=\{1\}$. A classical strongly regular graph is then  equivalent to an $(X;X)$-strongly regular graph.

 We first show that (part of) Proposition~\ref{prop1} can be adapted to this notion, as well, in some cases, and we deal below with one such instance.
 \begin{prop}
Let $G=(V,E,w)$ be  a weighted $(X;X)$-strongly regular graph of parameters $(v;r_0,r_1,\ldots,r_{q-1}; e_X,d_X)$, where $X\subseteq \mathbb{Z}_q$, $v=|V|$. Then,
\[
r_X(r_X-e_X-1)=d_X(v-r_X-1),
\]
where $r_X=\sum_{i\in X} r_i$.
\end{prop}
\begin{proof}
Without loss of generality we assume that the weights are additive, that is, they belong to $\Z_q$. Fix a vertex $\uu\in V$ and let $A$ be the set of vertices adjacent to $\uu$ with connecting edges of weight in $X$, and $B=V\setminus \{A,\uu\}$. Observe that $|A|=\sum_{i\in X} r_i=r_X$ and $|B|= v-r_X-1$. We somewhat follow the combinatorial method of the classical case, and we shall count the number of vertices between $A$ and $B$ in two different ways. For any vertex $\aa\in A$, there are exactly $e_X$ vertices in $A$ adjacent to both $\uu,\aa$ of edge weights in $X$, and so, exactly $r_X-e_X-1$ neighbors in $B$ whose connecting edges have weight in $X$. Therefore, the number of edges of weight in $X$ between $A$ and $B$ is $r_X(r_X-e_X-1)$.

 On the other hand, any vertex $\bb\in B$ is adjacent to $d_X$ vertices in $A$ of connecting edge with weight in $X$  (since $\uu,\bb$ must share $d_X$ common vertices of connecting edges of weight in $X$) and so, the total number of edges of weight in $X$ between $A$ and $B$ is $d_X(v-r_X-1)$. The proposition follows.
\end{proof}

Let $G=(V,E,w)$ ($w:E\to \Z_q$) be a weighted graph, where $w(E)\subseteq\Z_q$ (or $w(E)\subseteq \U_q$). We define the {\em complement} of $G$, denoted by $\bar G$ the graph of vertex set $V$ with an edge between two vertices $\aa,\bb$ having weight $q-1-f(\aa\+\bb)$ (or, multiplicatively, $\zeta^{q-1-f(\aa\+\bb)}$. This is a natural definition, since if $G$ is the Cayley graph associated to $f=a_0+2a_1, a_0,a_i\in\cB_n$, then we observe that $\bar G$ is the Cayley graph associated to $\bar f=\bar a_0+2\bar a_1+\cdots +2^{k-1} \bar a_{k-1}$, where $\bar a_i$ is the binary complement of $a_i$ (that follows from $2^k-1-f=(1-a_0)+2(1-a_1)+\cdots+2^{k-1} (1-a_{k-1})=\bar a_0+2\bar a_1+\cdots+2^{k-1} \bar a_{k-1}$).

\begin{lem}
\label{lem9}
Let $G=(V,E,w)$  $(w:E\to \Z_q)$ be a weighted regular graph of parameters $(v;r_0,r_1,\ldots,r_{q-1})$. Then the complement $\bar G$ is a weighted regular graph of parameters $(v;\bar r_0,\ldots,\bar r_{q-1})$, where $\bar r_{q-1-j}=r_{j}$.
\end{lem}
\begin{proof}
  Let $\aa$ be an arbitrary vertex.
Recall that  we denote by $N_j(\aa)$ the set of all neighbors of a vertex $\aa$ of corresponding edge weight $j$.
Since $G$ is weighted regular, then $|N_j(\aa)|=r_j$.  In the graph $\bar G$, the weight $j$ will transform into $q-1-j$, therefore $\bar r_{q-1-j}=r_j$ and the lemma is shown.
\end{proof}

Let $A\subset B$ and $x\in B$. As it is customary, we will denote by $x+A$ the set $\{x+a : a\in A\}$.

\begin{thm}
Let $G=(V,E,w)$  ($V\subseteq \BBF_2^n$, $w:E\to \Z_q$) be an $(X;Y)$-strongly regular, for some $X,Y\subseteq \Z_q$ with $|X|=2^{k-1}, q=2^k$, of parameters $(v;r_0,r_1,\ldots,r_{q-1};e_X,d_X)$ such that $q-1-X=X$ or $\bar X$, and $q-1-Y=Y$. Then, the complement $\bar G$ is a $(q-1-X;Y)$-strongly regular graph of parameters $(v;\bar r_0,\ldots,\bar r_{q-1};\bar e_{q-1-X},\bar d_{q-1-X})$, where $\bar r_{q-1-j}=r_{j}$,
$\bar e_{q-1-X}=e_X$ and $\bar d_{q-1-X}=d_X$, if $q-1-X= X$, respectively, $\bar r_{q-1-j}=r_{j}$, $\bar e_{q-1-X}=d_X$ and $\bar d_{q-1-X}=e_X$, if $q-1-X= \bar X$.
\end{thm}
\begin{proof}
 The first claim follows from Lemma~\ref{lem9}. We consider the two cases $q-1-X=X$, or $\bar X$, separately. As before, for any two vertices $\aa,\bb$ we denote by $N_Y(\aa,\bb)$ the set of all vertices $\cc$ adjacent to both $\aa,\bb$ such that $w(\aa,\cc)\in Y, w(\bb,\cc)\in Y$.

\noindent {\em Case $1$.} Let $q-1-X=X$. For any two vertices $\aa,\bb$ with $w(\aa,\bb)\in  X$, then
$
|N_Y(\aa,\bb)|=e_X,
$
since the weight of the edge between $\aa,\bb$ remains in $X$. Similarly, for two vertices $\aa,\bb$ with $w(\aa,\bb)\in  \bar X$, then $|N_Y(\aa,\bb)|=d_X$.

\noindent {\em Case $2$.} Let $q-1-X=X$. For any two vertices $\aa,\bb$ with $w(\aa,\bb)\in  X$, then the weight of the edge between $\aa,\bb$ in $\bar G$ is now in $\bar X$, and we know that in that case $N_Y(\aa,\bb)=d_X$. Similarly, for two vertices $\aa,\bb$ with $w(\aa,\bb)\in  \bar X$, then $|N_Y(\aa,\bb)|=e_X$.
\end{proof}

In the next theorem, we shall show a  strong regularity theorem (a Bernasconi-Codenotti correspondence) for gbents  $f\in\cGB_n^4$ when $n$ even and $k=2$.
For two vertices $\aa,\bb$ of the associated Cayley graph, for $i,j\in\{0,1,2,3\}$, let $N_{\{i,j\}}(\aa,\bb)$ be the set of all ``neighbor'' vertices $\ww$ to both $\aa,\bb$ such that the edges have additive weights $f(\ww\+\aa)\in\{i,j\},f(\ww\+\bb)\in\{i,j\}$.
\begin{thm}
\label{GB4}
Let $f\in\cGB_n^4$, $n$ even. Then $f$ is gbent if and only if the associated generalized Cayley graph is $(X;\bar X)$-strongly regular with $e_X=d_X$, for both $X=\{0,1\}$, and $X=\{0,3\}$,
 that is, if and only if the following two conditions are satisfied:
\begin{enumerate}
\item[$(i)$] For any two pairs of vertices $\{\aa,\bb\}$, $\{\cc,\dd\}$, 
then   $|N_{\{2,3\}}(\aa,\bb)|=|N_{\{2,3\}}(\cc,\dd)|$.
\item[$(ii)$] For any two pairs of vertices $\{\aa,\bb\}$, $\{\cc,\dd\}$, 
then  $|N_{\{1,2\}}(\aa,\bb)|=|N_{\{1,2\}}(\cc,\dd)|$.
\end{enumerate}
\end{thm}
\begin{proof}We know that $f=a_0+2a_1$, where $a_0,a_1\in\cB_n$, is gbent if and only if $a_1,a_1\+a_0$ are both bent (see~\cite{ST09,smgs}).
Let $\uu\in \V_n$. We have that:
\begin{enumerate}
\item $f(\uu)=0\Leftrightarrow a_0(\uu)=0,(a_1\+a_0)(\uu)=0$
\item $f(\uu)=1\Leftrightarrow a_0(\uu)=1,(a_1\+a_0)(\uu)=1$
\item  $f(\uu)=2\Leftrightarrow a_0(\uu)=0,(a_1\+a_0)(\uu)=1$
\item  $f(\uu)=3\Leftrightarrow a_0(\uu)=1,(a_1\+a_0)(\uu)=1$
\end{enumerate}

If $f$ is gbent, then $a_1,a_1\+a_0$ are both bent. Then, by \cite{BC1}, their respective graphs are srg with respective parameters $e=d,e'=d'$. We consider the following cases:
\begin{enumerate}
\item[$(a)$] Let any $\aa,\bb,\cc$ such that $f(\aa\+\cc)\in\{1,2\}$ and $f(\bb\+\cc)\in\{1,2\}$, then $(a_1\+a_0)(\aa\+\cc)=1=(a_1\+a_0)(\bb\+\cc)$. Since the graph corresponding to $a_1\+a_0$ is srg with $e'=d'$, then $|\{\cc:(a_1\+a_0)(\aa\+\cc)=1=(a_1\+a_0)(\bb\+\cc)\}|=e'$. Therefore, $|N_{\{1,2\}}(\aa,\bb)|=e'$.
\item[$(b)$] Let any $\aa,\bb,\cc$ such that $f(\aa\+\cc)\in\{2,3\}$ and $f(\bb\+\cc)\in\{2,3\}$, then $a_1(\aa\+\cc)=1=a_1(\bb\+\cc)$. Since the graph corresponding to $a_1$ is srg with $e=d$, then $|\{\cc:a_1(\aa\+\cc)=1=a_1(\bb\+\cc)\}|=e$. Therefore, $|N_{\{2,3\}}(\aa,\bb)|=e'$.
\end{enumerate}
Conversely, let the generalized Cayley graph be such that, for any two pairs of vertices $\{\aa,\bb\}$, $\{\cc,\dd\}$, 
then   $|N_{\{2,3\}}(\aa,\bb)|=|N_{\{2,3\}}(\cc,\dd)|$, and $|N_{\{1,2\}}(\aa,\bb)|=|N_{\{1,2\}}(\cc,\dd)|$. As seen in the first part of the proof, $|N_{\{2,3\}}(\aa,\bb)|=|\{\cc:a_1(\aa\+\cc)=1=a_1(\bb\+\cc)\}|$. This number is a constant, regardless of the value of $a_1(\aa\+\bb)$. This implies that the Cayley graph corresponding to $a_1$ is srg with $e=d$, where $e=|N_{\{2,3\}}(\aa,\bb)|$.

Similarly, $|N_{\{1,2\}}(\aa,\bb)|=|\{\cc:(a_1\+a_0)(\aa\+\cc)=1=(a_1\+a_0)(\bb\+\cc)\}|$. This number is a constant, regardless of the value of  $(a_1\+a_0)(\aa\+\bb)$. This implies that the Cayley graph corresponding to $a_1\+a_0$ is srg with $e'=d'$, where $e'=|N_{\{1,2\}}(\aa,\bb)|$. Since both $a_1$ and $a_1\+a_0$ are therefore bent, we conclude that $f$ is gbent.
\end{proof}

It is not hard to show that in some instances a ``uniform'' strong regularity will hold.
\begin{cor}
Let $S$ be a bent set $($see \textup{\cite{BK08}}$)$, that is, every element of $S$ is a bent function and the sum of any two such is also a bent function. Let $a_0,a_1\in S$. Then, the generalized edge-weighted Cayley graph of $f=a_0+2a_1$ is $(X;\bar X)$-strongly regular for any $X$ with $|X|=2$.
\end{cor}

\begin{rem}
One certainly could inquire whether a similar result holds for a gbent for $n$ odd.  Since the answer depends on a characterization (not currently known) of classical semibents in terms of their Cayley graphs, we leave that question for a subsequent project of an interested reader.
\end{rem}

While we cannot find a necessary and sufficient condition on a gbent in $\cGB_n^q$, $q=2^k$, we can follow a similar approach as in Theorem~\ref{GB4} to find a necessary condition on the Cayley graph of a generalized bent in $\cGB_n^q$. As in the previous result, for $X\subseteq \Z_q$ and two vertices $\uu,\vv$, let $N_X(\uu,\vv)$ be the set of vertices $\ww$ such that $f(\uu\+\ww)\in X$ and $f(\vv\+\ww)\in X$. As usual, $\bar \cc$ is the complement of the vector $\cc$, and for two vectors $\aa=(a_1,\ldots, a_t),\bb=(b_1,\ldots,b_t)$, the notation $\aa\preceq \bb$ means that $a_i\leq b_i$, for all $1\leq i\leq t$.
Recall that the  canonical injection  $\fun{\iota} {\V_s}{\Z_{2^s}}$,
$\iota(\cc) = \cc\cdot (1,2,\ldots,2^{s-1})=\sum_{j=0}^{s-1}c_j 2^j$, where $\cc=(c_0,c_1,\dots,c_{s-1})$.

\begin{thm}
Let $n$ be even, and $f=a_0+2a_1+\cdots+2^{k-1} a_{k-1}$, $k\geq 2$, $a_i\in\cB_n$, be a generalized Boolean function. If $f$ is gbent then the associated edge-weighted Cayley graph is $(X_\cc^0;X_\cc^1)$-strongly regular with $e_{X_\cc^0}=d_{X_\cc^0}$, where $X_\cc^i=\{ \iota(\tilde \cc)+\iota(\dd)\,:\,\tilde\cc\preceq (\cc,1), wt(\tilde \cc)\equiv i\pmod 2, \dd\preceq \bar \cc\}$, $i=0,1$,
for all $\cc\in\V_{k-1}$; that is, for all $\cc\in\V_{k-1}$, and for any two pairs of vertices $(\uu,\vv),(\xx,\yy)$,
\[
\left|N_{X_\cc^1} (\uu,\vv) \right|=\left|N_{X_\cc^1} (\xx,\yy) \right|.
\]
\end{thm}

\begin{proof} The weighted regularity of $f$ follows from Proposition~\ref{prop6}. If $f$ is gbent then by \cite[Theorem 8]{mmms17}, we know that
for each $\cc\in\V_{k-1}$, the Boolean function $f_\cc$ defined
as
\[ f_\cc(\xx) =
c_0a_0(\xx)\+c_1a_1(\xx)\+\cdots\+c_{k-2}a_{k-2}(\xx)\+a_{k-1}(\xx) \]
is a bent function  with $ \cW_{f_\cc}(\aa)= (-1)^{\cc\cdot \iota^{-1}( g(\aa))+s(\aa)}2^{\frac n2}$, for some
$\fun {g}{\V_n}{\Z_{2^{k-1}}}$, $\fun {s}{\V_n}{\F_{2}}$.

While we cannot control in a simple manner the Walsh-Hadamard spectra conditions of $f_\cc$ on the Cayley graph of a gbent $f$, we can derive some necessary conditions for $f$ to be gbent. Let $\cc\in\V_{k-1}$ and $f_\cc$ bent. Consider $\uu\in\V_n$. Certainly, the condition that $f_\cc(\uu)=1$ means that an odd number of   functions $a_j$, occurring (that is, the corresponding coefficient is nonzero) in $f_\cc$ will output 1 at $\uu$. The $a_j$'s corresponding to entries that are 0 in $\cc$ can be taken either 0 or 1 (hence the condition in the definition of $X_\cc^i$ that $\dd\preceq \bar \cc$). We see that the set of values of $f$ when $f_\cc(\uu)=1$  is exactly $X_\cc^1=\{ \iota(\tilde \cc)+\iota(\dd)\,:\,\tilde\cc\preceq (\cc,1), wt(\tilde \cc)\equiv 1\pmod 2, \dd\preceq \bar \cc\}$. Similarly, the set of values for $f$ when $f_\cc(\uu)=0$
is $X_\cc^0=\{ \iota(\tilde \cc)+\iota(\dd)\,:\,\tilde\cc\preceq (\cc,1), wt(\tilde \cc)\equiv 0\pmod 2, \dd\preceq \bar \cc\}$.

Since $f_\cc$ is bent, then  any two vertices, $\uu,\vv$, will have the same number of adjacent $\ww$ with
$f_\cc(\uu\+\ww)=f_\cc(\vv\+\ww)=1$, regardless of the value of $f_\cc(\uu\+\vv)$. This implies that $\left|N_{X_\cc^1} (\uu,\vv) \right|$ is constant for all $\uu,\vv$.
\end{proof}

 \section{Further comments}

We follow the notation of \cite{CJMPW} and define yet another strong regularity concept here.  Let $\Gamma$ be an edge-weighted graph (with no loops) with vertices $V$, edges $E$, and weight set $W$ (in~\cite{CJMPW}, $W$ was taken to be $\mathbb{Z}_q^*$, although it could be arbitrary). As before, for each $\uu\in V$ and $a\in{W}\cup\{0\}$, the weighted $a$-neighborhood of $u$, $N_a(\uu)$, is defined as follows:
\begin{itemize}
\item $N_a(\uu)=$ the set of all neighbors $\vv$ of $\uu$ in $\Gamma$ for which the edge $(\uu,\vv)\in E$ has weight $a$ (for each $a\in W$).
\item $N^0(\uu)=$ the set of all nonadjacent $\vv$ of $\uu$ in $\Gamma$ (i.e., the set of $\vv$ such that $(\uu,\vv)\notin E$), that is, $N^0(\uu)=V\setminus \cup_{a\in W} N_a(\uu)$. In particular, $\uu\in N^0(\uu)$.
\end{itemize}

In \cite{CJMPW}, the following definition of weighted strongly regular graph is given.
 Let $\Gamma$ be a connected edge-weighted graph which is regular as a simple (unweighted) graph. Let $W$ be the set of edge-weights of $\Gamma$. The graph $\Gamma$ is called an {\em edge-weighted local strongly regular} (to distinguish it from our definition we inserted the adjective ``local'') with parameters $v,\,k=(k_a)_{a\in W},\,\lambda=(\lambda_a)_{a\in W^3}$, and $\mu=(\mu_a)_{a\in W^2}$, denoted $SRG_W(v,k,\lambda,\mu)$, if $\Gamma$ has $v$ vertices, and there are constants $k_a,\,\lambda_{a_1,a_2,a_3}$, and $\mu_{a_1,a_2}$, for $a,\,a_1,\,a_2,\,a_3\in W$, such that
 \[
 |N_a(\uu)|=k_a \text{ for all vertices }\uu,
 \]
  and for vertices $\uu_1\neq \uu_2$ we have
 \[
 |N_{a_1}(\uu_1)\cap N_{a_2}(\uu_2)|=
 \begin{cases}
  \lambda_{a_1,a_2,a_3} & \text{ if }  \exists\, a_3\in W \text{ with  }\uu_1\in N_{a_3}(\uu_2);\\
  \mu_{a_1,a_2}& \text{ if } \uu_1\notin N_{a_3}(\uu_2) \text{ for all } a_3.
 \end{cases}
 \]

 As was observed in \cite{CJMPW} for functions $f: \F_{p^n}\rightarrow \F_p$, where several questions were posed, it is not clear what the connection between this concept and generalized (or $p$-ary) bentness is. Our strong regularity definition does allow us to show such a connection and in the case $k=2$, we have a complete Bernasconi--Codenotti correspondence~\cite{BC1,BCV}.

\end{document}